\documentclass[reqno,12pt]{amsart}

\usepackage{amsmath,amsthm,amscd,amsfonts,amssymb}
\usepackage{latexsym}
\usepackage{mathabx}

{

\newcommand{\labelnummer}{\mbox{\normalfont (\roman{numcount})}}%

\makeatletter

  {\let\curlabelspeicher\@currentlabel%
    \begin{list}{\labelnummer}%
      {\usecounter{numcount}\leftmargin0pt%
        \topsep0.5ex\partopsep2ex\parsep0pt\itemsep0ex\@plus1\p@%
        \labelwidth2.5em\itemindent3.5em\labelsep1em%
      }%
    \let\saveitem\item%
    \def\item{\saveitem%
      \def\@currentlabel{{\upshape\curlabelspeicher}$\,$\labelnummer}}%
    \let\savelabel\label%
    \def\label##1{\savelabel{##1}%
      \@bsphack%
        \ifmmode\else%
          \protected@write\@auxout{}%
          {\string\newlabel{##1item}{{\labelnummer}{\thepage}}}%
        \fi%
      \@esphack%
    }%
  }{\end{list}}%

\usepackage[usenames]{color}
\usepackage{comment}
\usepackage{appendix}
\newcommand{\CC}{\mathbb C}

\newcommand{\NN}{\mathbb N}

\newcommand{\R}{\mathbb R}

\newcommand{\Z}{\mathbb Z}

\newcommand{\mm}{\mathcal M}

\newtheorem{thm}{Theorem}[section]
\newtheorem{lemma}[thm]{Lemma}
\newtheorem{cor}[thm]{Corollary}
\newtheorem{prop}[thm]{Proposition}
\newtheorem{definition}[thm]{Definition}
\newtheorem{remark}[thm]{Remark}
\newtheorem{remarks}[thm]{Remarks}
\newtheorem{rem}[thm]{Remark}

\newcommand{\ov}[1]{\frac{1}{#1}}

\newcommand{{\bfe}}{{\bf{e}}}

\newcommand{\bel}{\begin{equation} \label}
\newcommand{\eeq}{\end{equation}}
\newcommand{\beq}{\begin{equation}}
\newcommand{\beqnn}{\begin{equation*}}
\newcommand{\eeqnn}{\end{equation*}}

\newcommand{\ba}{\begin{array}}
\newcommand{\ea}{\end{array}}
\newcommand{\bea}{\begin{eqnarray}}
\newcommand{\eea}{\end{eqnarray}}

\newcommand{\SCHR}{SCHR\"ODINGER }
\newcommand{\Schr}{Schr\"odinger }

\setlength{\parindent}{32pt}
\newcommand{\Ip}{{\rm Im }}
\newcommand{\Rp}{{\rm Re }}
\newcommand{\cF}{\mathcal{F}}
\newcommand{\lt}{\ell^{2}}
\newtheorem*{nota}{Notation}
\begin{document}

%
%
%
%

\title[Localization and quantum transport]{Dynamical localization and delocalization for random \SCHR operators with $\delta$-interactions in $\R^3$}

\author[P.\ D.\ Hislop]{Peter D.\ Hislop}
\address{Department of Mathematics,
    University of Kentucky,
    Lexington, Kentucky  40506-0027, USA}
\email{peter.hislop@uky.edu}

\author[W.\ Kirsch]{Werner Kirsch}
\address{Fakult\"{a}t f\"ur Mathematik und Informatik,  FernUniversit\"at in Hagen, 58097 Hagen, Germany}
\email{werner.kirsch@fernuni-hagen.de}

\author[M.\ Krishna]{M.\ Krishna}
\address{Ashoka University, Plot 2, Rajiv Gandhi Education City, Rai, Haryana
131029 India}
\email{krishna.maddaly@ashoka.edu.in}

\begin{abstract}
We prove that the random \Schr operators on $\R^3$ with independent, identically distributed random variables and single-site potentials given by $\delta$-functions on $\Z^3$, exhibit both
dynamical localization and dynamical delocalization with probability one. That is,
there are regions in the deterministic spectrum that exhibit dynamical localization, the nonspreading of wave packets, and regions in the deterministic spectrum where the models also exhibit nontrivial quantum transport, almost surely. These models are the first examples of ergodic, random \Schr operators exhibiting both dynamical localization and delocalization in dimension three or higher. The nontrivial transport is due to the presence of delocalized generalized eigenfunctions at positive energies $E > \pi^2$.  The general idea of the proof follows \cite{hkk3} in which lower bounds on moments of the position operator are constructed using these generalized eigenfunctions.  A new result of independent interest is a proof of the Combes-Thomas estimate on exponential decay of the Green's function for \Schr operators with infinitely-many $\delta$-potentials.
\end{abstract}

\medskip

\thanks{PDH is partially supported by Simons Foundation Collaboration Grant for Mathematicians No.\ 843327.}

\maketitle \thispagestyle{empty}


\tableofcontents

\section{Introduction: Dynamical delocalization}\label{sec:intro1}
\setcounter{equation}{0}

In a recent paper \cite{hkk3}, we explored the relationship between
the growth rate of generalized eigenfunctions $\psi_E$ of a discrete \Schr operator $H_V$ on $\ell^2(\Z^d)$ 
and the transport of the model as measured by moments of the position operator $X$. 
We applied these results to trimmed Anderson models. 
Following \cite{GKS,JSS}, we will say that a \Schr operator exhibits \emph{dynamical delocalization} if some moment of the position operator, evolved with the \Schr unitary group, increases without bound in time, for some initial condition in $\ell^2 ( \Z^d )$.

In this paper, we prove that a family of random \Schr operators (RSO) on $L^2 (\R^3)$, with random $\delta$-interactions located on $\Z^3$, exhibits dynamical delocalization with probability one. This results from the existence of bounded, generalized eigenfunctions at positive energies $E > \pi^2$. In a previous paper \cite{hkk2}, we proved dynamical localization for states with initial conditions energy-localized in certain subintervals of the deterministic spectrum below $\pi^2$. Hence, the random $\delta$-interaction models on $L^2 (\R^3)$ exhibit both dynamical delocalization and dynamical localization almost surely.



We measure quantum transport with respect to a weight function $\varphi$. The basic example is $\varphi_q(x) := \langle x \rangle^q = (1+\| x \|^{2})^{\frac{q}{2}}$, for $q > 0$. For $q=2$, this function measures the mean square displacement of a quantum wave packet $\psi_t :=e^{-itH} \psi$, with initial condition $\psi$.
We define the averaged $\varphi$-moment of the position operator $X$ for a time-evolving state with initial condition $\psi$ to be:
\begin{align}\label{defmm0}
   \mm_{\psi, \varphi}(T)~=~\ov{T}\int_{0}^{\infty} e^{-\frac{t}{T}}\,\langle \psi,
   e^{iHt} \varphi(X) e^{-iHt}\, \psi\rangle\;dt  .
\end{align}
For the choice of $\varphi_q$, we write ${\mathcal{M}}_{\psi;q}$ or simply ${\mathcal{M}}_q$ when the initial state $\psi$ isn't important for the discussion.


For random \Schr operators (RSO), Anderson localization (pure point spectrum with exponentially decaying eigenfunctions) can often be strengthened to dynamical localization. \textbf{Dynamical localization} means that the moments $\mathcal{M}_q(T)$ are bounded in time for all $q \geq 0$, almost surely, see, for example, \cite{GK-character}. On the other hand, if $\mathcal{M}_q(T)$ is unbounded as a function of $T$ for some $q$, almost surely, one speaks of \textbf{dynamical delocalization}.

For RSO on $\ell^2(\Z)$, dynamical delocalization was first proved for certain one-dimensional models, called random polymer models, by Jitomirskaya , Schulz-Baldes, and Stolz \cite{JSS}.
These models are characterized by the existence of a set of finite energies, called critical energies, at which the Lyapunov exponent vanishes (see De Bi\`evre and Germinet \cite{dBg2} for the dimer model, \cite{JSS}  for random polymer models, and Damanik, Sims, and Stolz \cite{DSS} for random word models). For the random polymer models, it is proved in \cite{JSS} that $\mathcal{M}_q(T)$ is bounded below by a positive power of $T$ verifying nontrivial quantum transport, almost surely, for an initial state supported at a single point.
A second family of one-dimensional models was studied by Drabkin, Kirsch, and Schultz-Baldes \cite{DKSB} who proved that the random Kronig-Penney model on $L^2 (\R)$, with a random potential formed from $\delta$-potentials on $\Z$, also exhibits dynamical delocalization. This model has a countable number of critical energies at which the Lyapunov exponent vanishes.
Finally, in two-dimensions, Germinet, Klein, and Schenker  \cite{GKS} proved that for the random Landau Hamiltonian on $L^2 (\R^2)$, the local transport exponent at an energy near each Landau level is strictly positive.

Hence, all three of these models exhibit dynamical delocalization. It is also known that these models exhibit dynamical localization for initial states in spectral subspaces corresponding to certain subintervals of the deterministic spectrum. For the one-dimensional models, these subintervals are away from the critical energies,
see \cite{dBg2} for the dimer model, and \cite{DSS} for the polymer models. For the two-dimensional random Landau operators, these subintervals are away from the  Landau levels near the band edges. To our knowledge, the random $\delta$-interaction model on $L^2(\R^3)$ is the first example of a three-dimensional, ergodic, RSO exhibiting both transport properties, dynamical localization and dynamical delocalization.



\subsection{Random delta interaction models (RDM)}\label{sec:intro1}

Schr\"odinger operators with point interactions are useful models of many quantum phenomena. In this article, we continue our study of random Schr\"odinger operators with potentials formed from delta interactions at the lattice points of $\Z^d$, $d=1,2,3$, and random coupling constants, see \cite{hkk1,hkk2}.
We study the formal Hamiltonian
\beq\label{eq:rdm1}
H_{\omega} =
- \Delta +  ~ \sum_{j \in \Z^3}  \omega_j \delta (x - j),
\eeq
on $L^2 (\R^3)$. The coupling constants $\{ \omega_j ~|~ j \in \Z^3 \}$ are a family of independent, identically distributed (iid) random variables with an absolutely continuous probability measure having a density $h_0 \in L_0^\infty (\R)$. Furthermore, we assume the support of $h_0$ is the interval $[-b, -a]$ for some finite constants $0 < a < b < \infty$.
We refer to \cite{aghkh} for a detailed discussion of the construction of these operators via the Green's function formulas given in \eqref{eq:green1} and \eqref{eq:kernel1}. We review this construction in section \ref{subsec:rdm1}.

The family of random \Schr operators \eqref{eq:rdm1} are covariant with respect to lattice translations. As a consequence, there exists a closed subset
$\Sigma \subset [-M, \infty)$, for some $0 \leq M < \infty$, so that the spectrum $\sigma (H_\omega) = \Sigma$, almost surely. Furthermore, under the conditions on the support of $h_0$, $\Sigma \cap \R^-$ is nonempty. We denote by
$\Sigma_{pp}$ the almost sure pure point component of $\Sigma$.
A main result of \cite{hkk1} is that the deterministic pure point spectrum
$\Sigma_{pp}$ is nonempty. This means that the spectrum is dense pure point almost surely in an interval in $\Sigma$ near $\inf \Sigma$. In \cite{hkk2}, we proved that if $E_0 \in \Sigma_{pp}$, then the local eigenvalue statistics exists and is given by a Poisson point process.

The RDM \Schr operator $H_{\omega}$ may be defined through its resolvent operator acting on $L^2 ( \R^3)$. As described in \cite[Theorem III.1.1.1]{aghkh}, the resolvent $R_\omega (z) := (H_{\omega} - z)^{-1}$ has an explicitly computable kernel, the Green's function $G_{\omega}(x,y;z)$, for $x \neq y \in \R^3$. To describe this kernel, we let $R_0(z) := (-\Delta - z)^{-1}$ be the resolvent of the Laplacian $H_0 := - \Delta$, with Green's function 
\begin{align}\label{eq:defG0}
   G_0 (x,y;z)~=~\frac{1}{4\pi}\,\frac{e^{i\sqrt{z}\| x-y \|}}{\| x-y \|}
\end{align}
where we take the principal branch for $\sqrt{z} $.

Then, the Green's function $G_\omega(x,y;z)$, corresponding to the resolvent $R_\omega ( z)$ of the RDM, is defined for $z \in \CC \backslash \Sigma$ and for $x \neq y$ by
 \beq\label{eq:green1}
G_{\omega}(x,y;z) := G_0(x,y;z) + \sum_{i,j \in \Z^3} \overline{G_0(x,i;{z})} [\Gamma (z, \omega)^{-1}]_{ij} G_0(j,y;z) .
\eeq
The infinity-by-infinite matrix kernel $\Gamma_{}(z, \omega)$ is the random matrix on $\ell^2 ( \Z^3)$ given by
\beq\label{eq:kernel1}
[\Gamma (z, \omega)]_{ij} = \left(  \frac{1}{\omega_{j}} - i e(z) \right) \delta_{ij} - G_0(i,j; z)(1 - \delta_{ij} ) ,
\eeq
where $e(z) =  \frac{\sqrt{z}}{4 \pi}$, taking the principal branch. The convergence of the interaction term involving the matrix $[\Gamma (z,\omega)]^{-1}$ is described in section \ref{subsec:ct1}.

\begin{rem}\label{rem:omega0}
   Equation \eqref{eq:kernel1} and hence \eqref{eq:green1} make only sense if $\omega_{i}\not=0$ for all $i\in\Z^{3} $. To allow $\omega_{i}=0$ in the formal expression \eqref{eq:rdm1} we define the resolvent of $H_{\omega}$ in such cases in the following way. Set $\Lambda=\{ i\in\Z^{3}\mid \omega_{i}\not=0 \} $ and define the Green's function $G_{\omega}(x,y;z)$ by
   \begin{align}
      G_{\omega}(x,y;z) := G_0(x,y;z) + \sum_{i,j \in \Lambda} \overline{G_0(x,i;{z})} [\Gamma_{\Lambda} (z, \omega)^{-1}]_{ij} G_0(j,y;z) .
   \end{align}
   where now $\Gamma_{\Lambda}(z, \omega)$ is the matrix \eqref{eq:kernel1} on $\ell^{2}(\Lambda) $.
\end{rem}

\subsection{Summary of the main results}\label{subsec:main1}

In our previous work \cite{hkk1}, we proved that the RDM has a region of pure point spectrum with probability one and exhibits dynamical localization in the same region under the following assumptions of the single-site probability measure $\mu$.

\medskip

\noindent
\textbf{Assumption 1}: \emph{The random variables $\{ \omega_j ~|~ j \in \Z^3 \}$
form a family of independent and identically distributed (iid) random variables with a common, absolutely continuous probability measure $\mu$. The density of $\mu$, denoted $h_0 \in L^\infty (\R)$, is supported in a set of the form $[-b, -a] \cup [f,g]$, for $0 < a < b < \infty$ and $0 < f < g < \infty$. }

\medskip

\begin{thm}\cite[Theorem 1.2]{hkk1}\label{thm:dyn_loc1}
Let $H_\omega$ be the random \Schr operators with $\delta$-interactions formally
given in \eqref{eq:rdm1}, and defined using the resolvent formula \eqref{eq:green1}. We assume that the random variables satisfy Assumption 1. Then, there exists a constant $E_1(a) > 0$ so that $\Sigma \cap ( -\infty, -E_1(a) ] \neq \emptyset$, and pure point almost surely, with exponentially decaying eigenfunctions. The family of operators exhibits dynamical localization in the Hilbert-Schmidt norm at all orders $q \in \NN$ on any interval  $I \subset \Sigma \cap ( -\infty, E_1(a) ]$,  That is, with probability one,
\beq\label{eq:dyn_loc1}
\sup_{t > 0} \| \langle x \rangle^\frac{q}{2} e^{-i H_\omega t} P_\omega (I) \psi \|_{HS} < \infty,
\eeq
for any initial state $\psi \in L^2 (\R^d)$, for $d = 1, 2, 3$.
The integrated density of states $N(E)$ is Lipschitz continuous for $E < 0$.
\end{thm}

In the present article, we prove that the RDM also exhibit dynamical \emph{de}localization.
We state this in part 1 of the following theorem. We also state the dynamical localization result (Theorem \ref{thm:dyn_loc1}) as part 2.

\begin{thm}\label{thm:loc_deloc1}
Let $H_{\omega}$ be the random $\delta$-interaction model described in \eqref{eq:rdm1} with iid random variables satisfying Assumption 1.
\begin{enumerate}
\item Dynamical delocalization: For any bounded interval $I \subset (\pi^2, \infty)$, there is a function $f_{I} $ such that  $P_\omega (I) f_{I} \neq 0$ and for
   $q>3$
\begin{align}\label{eq:dyndeloc1}
      \mm_{P_\omega (I) f_{I}, q}(T)~\geq~C(q, I)\, T  ,
   \end{align}
 for a positive constant $C(q,I)$ almost surely.

\medskip

\item Dynamical localization: There exists a nonempty interval $I \subset \Sigma$ so that $I \subset \Sigma_{loc}$ and for all $\psi \in L^2(\R^3)$, there is a finite constant $C \geq 0$, depending on $P_\omega (I) \psi$, so that
 \begin{align}
      \mm_{P_\omega (I) \psi, q}(T)  \leq C < \infty .
   \end{align}
\end{enumerate}
\end{thm}

 \medskip
 To our knowledge, this is the first example of a three-dimensional ergodic, random \Schr operator that exhibits both dynamical localization in one energy regime, and dynamical delocalization in another regime. Although we proved that the region of dynamical localization is a (possibly a subset of) the spectral localization regime, we do not know the nature of the deterministic spectrum in the region $(\pi^2, \infty)$.

 \begin{remarks}
 \begin{enumerate}
 \item The delocalization result is a deterministic result as can be seen from the ingredients of the proof.  In fact, the delocalization result holds independent of the
configuration of the coupling constants $\{ \omega_j \}$, including the \Schr operator  formally defined in \eqref{eq:rdm1} with some or all $\omega_j = 0$ (see Remark \ref{rem:omega0}).

\medskip

 \item We expect that a similar delocalization result holds for the RDM on $\R^2$.  This will be discussed in a separate article.  Dynamical localization was proved in \cite{hkk1}.

\medskip

\item The fact that \eqref{eq:dyndeloc1} holds for $P_\omega (I) f_{I}$,
and not just for the initial condition $f_{I}$, indicates that the spectrum in the interval $[ \pi^2, \infty)$ is responsible for the nontrivial transport.

\medskip

\item As can be seen from the proof the function $f_{I} $ can be chosen in the following way.
For any bounded interval $I\subset (\pi^{2},\infty)$ choose $r=r_{I}>0$ small enough and set $f_{I}(x)=\chi_{B_{r}(x_{0})}$ the characteristic function of the ball $B_{r}(x_{0}) $ of radius $r$ around the point $x_{0}=(\frac{1}{2},0,0) $.

\item For the one-dimensional RDM, the Kronig-Penney model, the spectrum is known to be almost surely pure point with exponentially decaying eigenfunctions. There are also critical energies $(k \pi)^2$, with $k \in \NN$, at which the Lyapunov exponent vanishes. The main result of \cite{DKSB} is that the diagonal of the integral kernel of the transport operator, $\mm_{q}(T)(a,a)$, for any $a \in \R \backslash \Z$, is bounded below by a positive power of $T$. Unfortunately, this result does not imply that there exists a function $\psi \in L^2 (\R)$ for which $\mm_{\psi, q}(T)$is unbounded in $T$.

\medskip

\item The importance of the threshold $\pi^2$ may be seen by considering a periodic $\delta$-interaction potential for which $\omega _j = \lambda$ for all $j \in \Z^3$. According to \cite[Theorem III.1.4.5]{aghkh}, the spectrum of $H_\lambda$ is
\beq\label{eq:periodic1}
\sigma (H_\lambda) = [ E_0^{\lambda} (0), E_0^{\lambda} (-  (\pi,\pi,\pi)) ] \cap [ E_1^\lambda , \infty ) ,
\eeq
where $E_0^\lambda (k)$ is the first band function as a function of $k \in B_0$, the Brillouin zone $B_0 := [- \pi, \pi ]^3 $. As $\lambda \to \infty$, the bottom of the second, semi-infinite band, approaches $\pi^2$. See \cite[sections 4.1-4.2]{hkk1} for further details.

 \end{enumerate}
  \end{remarks}

 Some ideas in this paper are related to the authors' recent paper \cite{hkk3}. In that work, we showed that dynamical delocalization follows from certain growth bounds on generalized eigenfunctions.  The main application of these results is to certain $\Gamma$-trimmed \Schr operators for which we proved dynamical delocalization. Related results to $\Gamma$-trimmed discrete \Schr operators, including Anderson models, can be found in \cite{CR,EK,ES,KiKr1,KiKr2}.

There are many works relating spectral properties of \Schr operators to quantum dynamics. We mention the papers \cite{bcm, combes, guarneri, GuarneriSB, JSB,last}.
Other works that relate the behavior of solutions to the \Schr equation and transport  and spectral properties include \cite{CKL,KL}.
 A detailed analysis of the metal-insulator transition, based on multi-scale analysis, is presented in \cite{GK-character}.

\section{The ingredients in the proof of dynamical delocalization for the RDM}
\label{sec:ingred1}
\setcounter{equation}{0}

In this section, we present the main aspects of the proof of the first part of Theorem \ref{thm:loc_deloc1}: section \ref{subsec:transport1}: Measure of transport; section \ref{subsec:rdm1}: \Schr operators with random $\delta$-interactions;
section \ref{subsec:ct1}: Combes-Thomas estimate for RDM; and section \ref{subsec:gen_ef1}: Generalized eigenfunctions.

\subsection{Measure of transport}\label{subsec:transport1}

We begin with the following theorem that will allow us to convert estimates on matrix elements of time evolved states to estimates on matrix elements of the resolvent.

\begin{definition}\label{defn:weight1}
 A nonnegative function $\varphi$ is a \emph{weight function} if it is monotonically increasing and $\lim_{|x| \to \infty} \varphi(x) = + \infty$. We write $\varphi_q$ for the canonical weight $\varphi_q(x) := \langle x \rangle^q$, for $q > 0$.
\end{definition}


\begin{thm}\label{thm:timeres} Let $X$ denote the position operator on $L^2 (\R^3)$, and let $\varphi$ be a weight as in Definition \ref{defn:weight1}. For $T>0$, and an initial state $\psi \in L^2(\R^3)$, we have
\bea
   \mm_{\psi; \varphi}(T)  & : = & \frac{1}{T} \int_0^\infty e^{- \frac{t}{T}} \langle e^{-itH} \psi, \varphi (X) e^{-itH} \psi \rangle  \label{eq:mm0} \\
   &  = &  \frac{1}{2\pi T}\,\int_{-\infty}^{\infty} \, \| \varphi^{\frac{1}{2}}
    R_H \left( E +\frac{1}{2T}i \right) \psi   \|^{2} \,dE .  \label{eq:mm1}
\eea
\end{thm}
A similar result is cited and used in \cite[section 3]{JSB} and in \cite[Lemma 6.3]{GK-character}. We present a proof based similar to the one in \cite[Lemma 6.3]{GK-character}.

\begin{proof}
We begin with an identity for any $T > 0$ and any self-adjoint operator $H$:
\beq\label{eq:resolv1}
-i R_H(E + i \frac{1}{2T} ) = \int_0^\infty e^{-i t( H - i\frac{1}{2T} - E)} ~dt.
\eeq
The integral is absolutely convergent. It follows that
\beq\label{eq:resolv2}
-i \varphi^{\frac{1}{2}} R_H( E + i \frac{1}{2T}) \psi =  \int_0^\infty e^{itE} g(t;x) ~dt,
\eeq
where the vector $g(t;x)$ is given by
\beq\label{eq:vector1}
g(t;x) = \varphi (x)^{\frac{1}{2}} e^{-i t( H - i\frac{1}{2T})} \psi(x), t \geq 0 .
\eeq
We extending the $L^2$-function $g(t;x)$ to be zero for $t < 0$.
Denoting by $\widehat{g}(E;x)$ the partial Fourier transform with respect to $t$, we obtain from the Plancherel identity
\beq\label{eq:plancherel1}
\frac{1}{2 \pi} \int_\R | \widehat{g}(E;x)|^2 ~dE = \int_0^\infty |g(t;x)|^2 ~dt .
\eeq
Integrating over $x \in \R^3$, we obtain from \eqref{eq:plancherel1}
\beq\label{eq:plancherel2}
\frac{1}{2 \pi} \int_\R \| \varphi^{\frac{1}{2}}  R_H( E + i \frac{1}{2T}) \psi \|^2 ~dE
= \int_0^\infty e^{- \frac{t}{T}} \| \varphi^{\frac{1}{2}}  e^{-it H} \psi \|^2 ~dt .
\eeq
Returning to the definition of $\mm_{\varphi;\psi}(T)$ in \eqref{eq:mm0},
the identity \eqref{eq:plancherel2} allows us to conclude that
\beq\label{eq:mm_final1}
\mm_{\psi;\varphi}(T) = \frac{1}{2 \pi T}  \int_\R \| \varphi^{\frac{1}{2}} R_H( E + i \frac{1}{2T}) \psi \|^2 ~dE,
\eeq
proving the identity \eqref{eq:mm1}.
\end{proof}`


\subsection{\Schr operators with random $\delta$-interactions}\label{subsec:rdm1}

We refer the reader to the encyclopedic work, \emph{Solvable models in quantum mechanics} \cite{aghkh}, for the construction of, and spectral properties of, the \Schr operators with $\delta$-interactions in dimensions 1, 2, and 3. These are the only dimensions for which the $\delta$-interaction is well-defined. We summarize and extend some of these results needed in our work for $\delta$-interactions on $\R^3$.

In \cite[III.1, Theorem 1.1.1]{aghkh}, the authors prove that the representation \eqref{eq:green1}, with $\Gamma (z, \omega)$ given in \eqref{eq:kernel1}, is well-defined and represents the Green's kernel of the self-adjoint operator $- \Delta_{\omega, \Z^3}$ (the notation of \cite{aghkh} for the operator formally given in \eqref{eq:rdm1}). They establish this result for ${\rm Im} ~k > 0$ \emph{sufficiently large.}

We analyze the interaction kernel $\Gamma (z, \omega)$ in Appendix \ref{subsec:Gamma1}.  We prove that it is a dissipative operator, that is, $- {\rm Im} ~ \Gamma (z, \omega)$ is bounded below by ${\rm Im} ~k > 0$. This allows us to prove that the inverse exists as a bounded operator on $\ell^2(\Z^3)$ for all ${\rm Im} ~k > 0$ and is analytic on $\CC^+$.
Since this operator agrees with the one constructed for ${\rm Im} ~k > 0$ sufficiently large, and is analytic in $k$, for ${\rm Im} ~ k > 0$, so by Vitali's Theorem it provides an extension of the result of \cite[III.1,Theorem 1.1.1, (1.1.9)]{aghkh} to the entire upper half plane ${\rm Im} ~k > 0$. We prove that the matrix elements $[ \Gamma^{-1}(z, \omega)]_{ij}$ decay exponentially in $|i-j|$, see \eqref{eq:gamma_inv2}. Together with the exponential decay of the free Green's function $G_0$, this establishes that the formula \eqref{eq:green1} is well-defined for all ${\rm Im} ~k > 0$.



\subsection{The Combes-Thomas estimate for the random $\delta$-potential \Schr operators}\label{subsec:ct1}

The Combes-Thomas estimate on the resolvent of $H_\omega$ plays a central role in our analysis. We prove a Combes-Thomas estimate on the resolvent of the singular RSO $H_{\omega}$ localized between two localized functions. The main technical result is the following bound on the interaction term \eqref{eq:kernel1}  in the expression of the resolvent. Let $C_0 \subset \R^3$ be the unit cube centered at the origin. We recall that $H_0 := - \Delta$, the nonnegative Laplacian on $L^2 (\R^3)$, has Green's function $G_0(x,y;z)$ given explicitly in \eqref{eq:defG0}. 
For $x\in\R^{3}$ we set
\begin{align}\label{eq:def_d}
   d(x)~=~\text{dist}(x,\Z^{3})~=~\inf \{\, \| x-m \|\mid m\in\Z^{3} \}
\end{align}
The following theorem is the main technical result and the Combes-Thomas estimate is given in Theorem \ref{thm:ct_final}.

\medskip

\begin{thm} \label{thm:ct_1}
Let $x,y \in \R^3\setminus \Z^{3}$, and $z \in \CC$ with ${\rm Im} ~ z > 0$.
There exist finite constants $\widetilde{C}=\widetilde{C}_{z}, {\widetilde{\gamma}} (z) > 0$, so that
    \begin{enumerate}
    \item The interaction kernel $\Gamma (z, \omega)$ in \eqref{eq:kernel1} satisfies
\beq\label{eq:interact_term1}
\left| \sum_{m,n \in \Z^3} \overline{G_0 (x,m;z)} [\Gamma (z,\omega)^{-1}]_{mn} G_0 (n,y;z) \right| \leq \frac{\widetilde{C}}{d(x)\,d(y)} e^{- \widetilde{\gamma} (z) \|x-y \|} .
\eeq

\medskip

\item We have
\beq\label{eq:interact_term_ave1}
\int_{C_0} ~ d^3x \int_{C_0} d^3y
\left| \sum_{j, \ell \in \Z^3} \overline{G_0 (x +m, j ;z)} \Gamma^{-1}(j, \ell) G_0 (\ell, y + n  ;z) \right| \leq {\widetilde{C}} e^{- \widetilde{\gamma} (z) \|n - m \|} .
\eeq
\end{enumerate}

\end{thm}

We present the proof of this theorem in Appendix \ref{app: ct_proof1}.
The Combes-Thomas estimate for the Hamiltonian $H_{\omega}$ follows from Theorem \ref{thm:ct_1}. 

\begin{thm}\label{thm:ct_final} Let $G_{\omega}(x,y;z)$ denote the kernel of the resolvent of $H_{\omega}$ as in
\eqref{eq:green1}.
   For  $z\in\CC$ with ${\rm Im}\,z>0 $, there exist finite, positive constants $C=C_{z}$ and $\gamma(z)$, such that
       \begin{enumerate}
    \item For $x,y\in\R^{3} $ with $\| x-y \|\geq 1 $ we have

\beq\label{eq:CT-final1}
\left| G_{\omega}(x,y;z) \right| \leq \frac{C}{d(x)\,d(y)} e^{- \gamma (z) \| x-y \|} ,
\eeq
where $d(x)$ is defined in \eqref{eq:def_d}. 

\medskip
   \item
   For all $m,n\in \Z^{3}$
   \begin{align}\label{eq:CT-final2}
      \int_{C_0} ~ d^3x \int_{C_0} d^3y
\left| {G_{\omega} (x +m, y+n ;z)}  \right| \leq {C} e^{- {\gamma} (z) \|n - m \|}
   \end{align}
   \end{enumerate}
\end{thm}

\begin{proof}
By \eqref{eq:green1}, we have for $\| x-y \|\geq 1 $
\begin{align}\label{eq:greenabs}
|G_{\omega}(x,y;z)|~ &\leq~ |G_0(x,y;z)| + |\sum_{i,j \in \Z^3} \overline{G_0(x,i;{z})} [\Gamma (z, \omega)^{-1}]_{ij} G_0(j,y;z)|\notag \\
&\leq~\frac{1}{4\pi}\,e^{-\tau(z) \|x-y \|}~+~\frac{\widetilde{C}}{d(x)d(y)}\,
e^{-\widetilde{\gamma}(z) \|x-y \|}\notag\\
&\leq~ \frac{C}{d(x)\,d(y)}\,e^{-\gamma(z) \|x-y \|}  ,
\end{align}
where we used \eqref{eq:defG0}, the notation $\tau(z) :={\rm Im}(\sqrt{z})$, $\gamma(z) :=\min(\tau(z),\widetilde{\gamma}(z))$, and Theorem \ref{thm:ct_1}.
Part 2 follows by integration.
\end{proof}


\subsection{Generalized eigenfunctions}\label{subsec:gen_ef1}

For any ${\bfe} > 0$, we construct bounded, generalized eigenfunctions $\psi_E$ of the operator $H_\omega$ with eigenvalue $E = {\bfe} + \pi^2$.
Let $\rho \geq 0$ be a bounded function with compact support in $[ - \sqrt{{\bfe}} / 2 , \sqrt{{\bfe}} / 2]$ normalized so that $\int \rho = 1$.   We construct the function of $(u,v)\in\R^{2}$ by
\beq\label{eq:ef1}
\psi_{0;{\bfe}}(u,v) := \int_\R e^{i s u} e^{i v \sqrt{ {\bfe} - s^2}} ~ \rho (s) ~ ds ,
\eeq
so that $ - \Delta \psi_{0;{\bfe}} = {\bfe} \psi_{0;{\bfe}}$ with $\psi_{0;{\bfe}}(0,0) = 1$. 
We note the following two properties of $\psi_{0;\bfe}$:

 \begin{align}\label{eq:gef_z1}
 \int_{\R} | \psi_{0;{\bfe}}(u,v)|^2 ~du \leq C_1 <\infty \quad\text{and} \quad
\int_{\R} | \psi_{0;{\bfe}}(u,v)|^2 ~dv \leq C(\bfe)<\infty  .
  \end{align}

 For any $E =:  \pi^2 + {\bfe}$, and ${x} := (x_1,x_2,x_3) \in \mathbb{R}^3$, we define the bounded, generalized eigenfunctions  $\psi_E$ of $H_\omega$ by
\beq\label{eq:gef2}
\psi_E (x) := \sin ( \pi x_{1}) \psi_{0;{\bfe}}(x_{2},x_{3} )  = \sin ( \pi x_{1}) \int_\R e^{i s x_{2}}
e^{i x_{3} \sqrt{ {\bfe} - s^2}} ~ \rho (s) ~ ds  .
\eeq
Note that $\psi_E | \Z^3 = 0$ and
 $$
H_0 \psi_E := - \Delta \psi_E  = (  \pi^2 + {\bfe} ) \psi_E  = E \psi_E .
 $$
 As a consequence, we formally have
 \beq\label{eq:gev1}
 H_\omega \psi_E = (  \pi^2 + {\bfe}) \psi_E = E \psi_E.
  \eeq
 We will refer to $\psi_E$ as a generalized eigenfunction of $H_\omega$. Since $\psi_E \not\in L^2 (\R^3)$, we don't use the eigenvalue equation \eqref{eq:gev1} but the version in Lemma \ref{lemma:gef1} for a truncated form of $\psi_E$.

We consider a smooth, nonnegative bump function $\chi$ on $\mathbb{R}$ which
is $1$ in $[-1, ~1]$, supported in
$[-2, ~2]$. For $L > 0$, we define $\chi_L(x_{1},x_{2},x_{3}) = \chi(\frac{x_{1}}{L})\cdot \chi(\frac{x_{2}}{L}) $.
We note that $\chi_L$ is compactly supported in $(x_1, x_2)$ and is independent of $x_3$. 

We then get a family of `approximate' generalized eigenfunctions of $H_\omega$, given by
\begin{equation}\label{eqnK1}
\psi_{L,E}(x) = \chi_{L}(x)\,\psi_E(x)\,.
\end{equation}
Each function $\psi_{L,E}$ is bounded with $\| \psi_{L,E}\|_\infty \leq 1$, and 
$\psi_{L,E}\in L^{2}(\R^{3}) $ due to \eqref{eq:gef_z1}, with norm $\mathcal{O}(L^2)$.

\medskip

\begin{lemma}\label{lemma:gef1}
The family of bounded approximate generalized eigenfunctions $\psi_{L,E}$
in $\eqref{eqnK1}$, satisfy the following relation:
\bea\label{eq:gef3}
H_\omega \psi_{L,E} & =  & H_0 \chi_L \psi_E   \nonumber \\
 &= & \chi_L H_0 \psi_E + [H_0, \chi_L] \psi_E \nonumber \\
  &= & ( \pi^2 + {\bfe}) \psi_{L,E} + [H_0, \chi_L] \psi_E .
\eea
\end{lemma}

\begin{proof}
The function $\psi_{L,E} \in D(H_0)$ and $\psi_{L,E} |_{\Z^3} = 0$,
by construction of $\psi_E$.
According to \cite[III.1, Theorem 1.1.2]{aghkh}, $\psi_{L,E}$ is in the domain $D(H_{\omega}) $ of $H_{\omega}$ and $H_{\omega}\psi_{L,E}=H_{0}\psi_{L,E}$
\end{proof}
\medskip

\noindent We need estimates on the commutator term in \eqref{eq:gef3}.
We observe
$$\int_{\R^3} |\psi_{L,E}(x)|^2 ~d^3 x  = \mathcal{O}(L^2)\,,$$
since $\psi_{L,E} $ is $L^{2}$ in $x_{3}$-direction.

The partial derivatives, needed for the commutator estimate, satisfy
 $$
 \int_{\R^3} \left| \left( \frac{\partial}{\partial w} \chi_L(x) \right) \left( \frac{\partial}{\partial w}  \psi_{E}(x) \right) \right|^2 ~d^3 x \leq C <\infty \,.
 $$
for $w \in\{ x_1, x_2, x_3 \}$.

These estimates lead to the following bounds on the commutator in \eqref{eq:gef3}.

\begin{lemma}\label{lemma:remainder1}
There exist finite constant $A_1, A_2 = A_2(E) > 0$, independent of $L$, so that
\beq\label{eq:comm_bd2}
\int_{\R^3} | ( \Delta {\chi_L}({x}) )  \psi_{E}({x})|^2 ~d^3 {x} \leq A_1  \,,
\eeq
and
\beq\label{eq:comm_bd1}
\int_{\R^3} | \nabla {\chi_L}({x}) \cdot \nabla \psi_E({x})|^2 ~d^3 {x} \leq A_2 \,.
\eeq
\medskip
As a consequence, the function
\bea\label{eq:error1}
\mathcal{C}_{L;E}({x})  & := & ([ - \Delta, \chi_L] \psi_E )({x})  \nonumber \\
 &= & - 2 \{ ( \nabla \chi_L)  \cdot (\nabla \psi_E) \} ({x}) - \{ ( \Delta \chi_L ) \psi_E \} ({x}) ,   \nonumber   \\
\eea
satisfies the bound
\beq\label{eq:unif_error_bd1}
\| \mathcal{C}_{L;E} \|_{L^2(\R^3)} \leq A_0 (E) \, ,
\eeq
where the finite constant $A_0(E)$ is locally bounded in $E$, and independent of $L$.
\end{lemma}

 Because of Lemma \ref{lemma:gef1} and Lemma \ref{lemma:remainder1}, the family  $\{ \psi_{L,E} \}_{L}$ is a type of a `generalized Weyl sequence' for $H_\omega$ and $E= \pi^2 + {\bfe}$, any ${\bfe} > 0$. For $z = E + i \epsilon$, with $\epsilon > 0$, we obtain
 \beq\label{eq:weyl1}
 (H_\omega - z  ) \psi_{L,E} = - i \epsilon \psi_{L,E}  +\mathcal{C}_{L;E}. 
 \eeq
 By the bound \eqref{eq:unif_error_bd1}, we obtain for $\epsilon = 0$, that
 \beq\label{eq:weyl2}
 \|  (H_\omega - E ) \psi_{L,E}  \| \leq A_0(E) \, .
 \eeq
We remind the reader that $\sup \| \psi_{L,E} \|_{2} $ is \emph{not} bounded with respect to $L$ however, a property required for an honest Weyl sequence.

\section{Proof of the main theorem: Dynamical delocalization}\label{sec:dynam_deloc1}
\setcounter{equation}{0}

Let  $x_0 := ( \frac{1}{2}, 0, 0)$ and $\chi_{B_t(x_0)}$ be the characteristic function of a ball ${B_t(x_0)} \subset \R^3 \backslash \Z^3$, for $0 < t < \frac{1}{2}$. Let $\psi_E$ be a generalized eigenfunction of $H_\omega$ with eigenvalue $E$ as in \eqref{eq:gef2}.
We take $\psi_{L,E}$ as in equation (\ref{eqnK1}).

For $z=E  + i \epsilon \not\in \sigma (H_\omega)$, for $\epsilon > 0$, as in Lemma \ref{lemma:gef1}, we compute
\bea\label{eq:resolv_loc1}
\psi_{L,E}  & = & R(z) (H_\omega - z)\psi_{L,E} =  R(z) (H_\omega - z) \chi_L \psi_E \nonumber \\
 & = & -i \epsilon R(z) \psi_{L,E} +   R(z)[ -\Delta, \chi_L] \psi_E \nonumber \\
  & = &    -i \epsilon R(z) \psi_{L,E} +   R(z) \mathcal{C}_{L;E}  ,
 \eea
 where the $L^2(\R^3)$-function $\mathcal{C}_{L;E} := [ -\Delta, \chi_L] \psi_E$ , as  in \eqref{eq:error1}.

Multiplying the equation \eqref{eq:resolv_loc1} on the left by $\chi_{B_t(x_0)}$,  integrating, and expressing the finite integrals as inner products, we obtain 
  \bea\label{eq:resolv_loc2}
\langle \chi_{B_t(x_0)},  \psi_E \rangle
  & = &  -i \epsilon \langle \chi_{B_t(x_0)} ,  R(z) \psi_{L,E} \rangle  + \langle \chi_{B_t(x_0)},  R(z)[ -\Delta, \chi_L] \psi_E \rangle \nonumber \\
   & = &  -i \epsilon \langle \chi_{B_t(x_0)} ,  R(z) \psi_{L,E} \rangle  +
   \langle \chi_{B_t(x_0)},  R(z) \mathcal{C}_{L;E} \rangle .
 \eea
We first estimate the matrix element involving  $\mathcal{C}_{L;E}$.

 \medskip

 \begin{lemma}\label{lemma:ct_app1}
 For any fixed $\epsilon > 0$, there is an $L_\epsilon > 0$, so that for all $L > L_\epsilon$, we have
\beq\label{eq:resolv_loc3}
|\langle \chi_{B_t(x_0)},  R(z) \mathcal{C}_{L;E}  \rangle | \leq \frac{1}{4}
| \langle \chi_{B_t(x_0)},  \psi_E \rangle  | ,
\eeq
where  $\mathcal{C}_{L;E} := [ -\Delta, \chi_L] \psi_E .$
\end{lemma}

\begin{proof}
Recalling that $z = E + i \epsilon$, and that fact that $\mathcal{C}_{L;E}$ is supported in $\mathcal{A}_L := \Lambda_{2L}(0) \backslash \Lambda_L(0)$, we obtain using Theorem \ref{thm:ct_final}
\bea\label{eq:me_bd1}
| \langle \chi_{B_t(x_0)},  R(z) \mathcal{C}_{L;E} \rangle |  & \leq & C_z e^{- \gamma(z) L} \left( \int \chi_{B_t(x_0)} ({x}) ~d^3 {x} \right) \left( \int_{\mathcal{A}_L} \mathcal{C}_{L;E} ({x})  ~ d^3 {x} \right) \nonumber  \\
 & \leq  &  C_z A_0(E)  e^{- \gamma(z) L} ,
 \eea
where the constant $A_0(E)$ is from the bound \eqref{eq:weyl2}. The constant $C_z$,  and the exponent $\gamma(z) > 0$, are from the Combes-Thomas estimate. It follows that for $\epsilon > 0$ fixed, and uniformly in $E$ belonging to a compact subset of $\R^+$, we can make the left side of \eqref{eq:me_bd1} as small as desired. In particular, for large $L > L_\epsilon$, inequality \eqref{eq:resolv_loc3} holds.
\end{proof}

\medskip

It follows from this and \eqref{eq:resolv_loc2} that for $L$ large enough, depending on $\epsilon$, we have
\beq\label{eq:resolv_loc4}
\frac{3}{4} | \langle \chi_{B_t(x_0)},  \psi_E \rangle  | \leq
 \epsilon | \langle \chi_{B_t(x_0)} ,  R(z) \psi_{L,E} \rangle  | .
\eeq
We next introduce the spectral projectors $P_\omega (I)$ for $H_\omega$ and any compact interval $I \subset (\pi^2 , \infty)$:  $ P_\omega (I) + P_\omega ({I^c}) = 1$. Inserting this identity after the resolvent on the right side of \eqref{eq:resolv_loc4}, we obtain two terms:
\bea\label{Eq:me1}
\langle \chi_{B_t(x_0)} ,  R(z) \psi_{L,E} \rangle   & = &
\langle \chi_{B_t(x_0)} ,  R(z) P_\omega (I) \psi_{L,E} \rangle \nonumber \\
 &  & +
 \langle \chi_{B_t(x_0)} ,  R(z) P_\omega ({I^c}) \psi_{L,E} \rangle  \nonumber \\
   & =: & A + B .
\eea
We analyze term B using Lemma \ref{lemma:proj_bd1} and \eqref{eq:proj_bd2}. In particular, in Lemma \ref{lemma:proj_bd1}, we take $E = \bfe + \pi^2$, $\psi$ to be $\psi_{L;E}$, $\delta = \frac{1}{2} |I|$,  and the constant $a$ in \eqref{eq:hyp1} to be $A_0(E)$ in \eqref{eq:weyl2}. As a result, we have from \eqref{eq:proj_bd2}
\bea\label{eq:est_B_1}
| B |  & = &   | \langle \chi_{B_t(x_0)} ,  R(z) P_\omega ({I^c}) \psi_{L,E} \rangle |
\nonumber \\
  & \leq &  C t^3 \frac{A_0(E)}{|I|^2} ,
  \eea
so the upper bound for $|B|$ is independent of $L$. Hence, we can take $\epsilon$ small enough so
$$
\epsilon |B| \leq \frac{1}{4} | \langle \chi_{B_t(x_0)},  \psi_E \rangle   |  .
$$
As a result, it follows from \eqref{eq:resolv_loc4} that  we have
\beq\label{eq:resolv_loc5}
\frac{1}{2} | \langle \chi_{B_t(x_0)},  \psi_E \rangle  | \leq
 \epsilon | \langle \chi_{B_t(x_0)} ,  R(z)  P_\omega ({I}) \psi_{L,E} \rangle  | .
\eeq

\medskip

We recall the moment weight function $\varphi_q(x) = \langle x \rangle^q$.
We introduce this into the matrix element on the right side of \eqref{eq:resolv_loc5}. By the Cauchy-Schwartz inequality we obtain:
\bea\label{eq:monent1}
| \langle \chi_{B_t(x_0)} ,  R(z)  P_\omega ({I}) \psi_{L,E} \rangle |
& = & | \langle \varphi_q ^{\frac{1}{2}}   P_\omega ({I}) R(z)^* \chi_{B_t(x_0)} ,  \varphi_q ^{-\frac{1}{2}}  \psi_{L,E} \rangle |  \nonumber \\
 & \leq & \| \varphi_q ^{\frac{1}{2}} P_\omega ({I}) R(z) \chi_{B_t(x_0)} \|_2 ~ \|\varphi_q ^{-\frac{1}{2}} \psi_{L,E} \|_2 \nonumber \\
  & \leq & \|\varphi_q ^{\frac{1}{2}}  R(z)  P_\omega ({I}) \chi_{B_t(x_0)}  \|_2 ~
  \| \varphi_q ^{-\frac{1}{2}} \psi_E \|_2 .  \nonumber \\
   & &
  \eea
  This yields the key lower bound:
  so that
  \beq\label{eq:resolv_loc6}
\frac{1}{2} \left(  \frac{ | \langle \chi_{B_t(x_0)},  \psi_E \rangle  |}{\| \varphi_q ^{-\frac{1}{2}} \psi_E \|_2}  \right)  \leq
 \epsilon \|\varphi_q ^{\frac{1}{2}}  R(z)  P_\omega ({I}) \chi_{B_t(x_0)}  \|_2
  .
\eeq

\medskip

%

In the following lemma, we provide bounds on the numerator and denominator of the left of \eqref{eq:resolv_loc6}.

\begin{lemma}\label{lemma:bounds2}
Let $\varphi_q(x) := \langle x \rangle^q$, for $q >0$.
\begin{enumerate}
\item  For $q > {3}$, we have
\beq\label{eq:norm_ub1}
\| \varphi_q^{-\frac{1}{2}} \psi_E \|^2 \leq  4 \pi \left(  \frac{q}{q -3}\right)  .
\eeq

\medskip

\item
Let $x_0 = ( \frac{1}{2}, 0 ,0)  \in \R^3 \backslash  \Z^3$ and $0 < t < \frac{1}{2}$.
Then, for $\bfe > 0$, there exists a radius $0 < t_E < \frac{1}{4}$, so that
\bea\label{eq:integral_lb1}
| \langle \chi_{B_{t_E}(x_0)},  \psi_E \rangle | & \geq & \frac{\sqrt{2}}{2} | B_{t_E}(x_0)| \cos (\frac{1}{2})   \nonumber \\
  & \geq & \frac{2\sqrt{2}}{3} \pi t_E^3 \cos (\frac{1}{2})  .
  \eea
\end{enumerate}
\end{lemma}

\begin{proof}
\noindent
1.  We recall the bounded, generalized eigenfunction $\psi_E(x)$, for $E = \pi^2 + \bfe$, with $\bfe > 0$, given in \eqref{eq:gef2}:
\beq\label{eq:gen_ef6}
\psi_E (x ) =  \sin ( \pi x_1) \int_\R 
e^{i s x_2} e^{i x_3 \sqrt{\bfe -s^2}} \rho(s)  ~d s ,
\eeq
with ${\rm supp} ~\rho = [- \frac{\sqrt{\bfe}}{2} , \frac{\sqrt{\bfe}}{2} ]$.
It follows from the normalization of $\rho$ that $| \psi_E(x)| \leq 1$. Consequently, we have
\beq\label{eq:gen_ef7}
\| \varphi_q^{-\frac{1}{2}} \psi_E \|_2^2 \leq 4 \pi \int_0^\infty \frac{r^2}{(r^2 +1)^{\frac{q}{2}}} ~dr .
\eeq
The estimate follows from this integral.

\noindent
2. As for part (2), we compute $\Rp ~ \psi_E$:
\beq\label{eq:gen_ef8}
\Rp ~ \psi_E(x )  =  \sin ( \pi x_1) \int_{- \frac{\sqrt{\bfe}}{2}}^{\frac{\sqrt{\bfe}}{2}} \cos (  sx_2 + x_3 \sqrt{\bfe-s^2} ) \rho (s) ~d s .
\eeq
The characteristic function $\chi_{B_{t_E}(x_0)}$ restricts the generalized eigenfunction $\psi_E$ in \eqref{eq:gen_ef6} to $B_t(x_0)$, where we take $x_0 = ( \frac{1}{2}, 0, 0)$, and $t \in (0, \frac{1}{4})$. 
It follows that $x_1 \in (\frac{1}{2} - t, \frac{1}{2}+ t)$ so that
\beq\label{eq:lb1}
\sin(\pi x_1) > \sin (\pi (\frac{1}{2} - t)) \geq \frac{\sqrt{2}}{2} , ~~~
0 < t < \frac{1}{4} .
\eeq
Provided the cosine function in \eqref{eq:gen_ef8} is strictly positive over the range of integration, it follows that
\bea\label{eq:integral_lb2}
| \langle \chi_{B_{t_E}(x_0)},  \psi_E \rangle | & \geq &
\int_{B_{t_E}(x_0)} \Rp ~ \psi_E({x} )  ~d^3 {x}  \nonumber \\
 & \geq & \left( \frac{ \sqrt{2}}{2} \right)  ~ \int_{{B}_{t_E}(x_0)} ~
\left[  \int_{- \frac{\sqrt{\bfe}}{2}}^{\frac{\sqrt{\bfe}}{2}}
 \cos (  s x_2 + x_3 \sqrt{\bfe-s^2} ) \rho(s)  ~d s \right]  ~ d^3 x   \nonumber \\
   & &
\eea
We let $\Theta : = sx_2 + x_3 \sqrt{\bfe-s^2}$, for $x \in B_{t_E}(x_0)$ and $s \in (- \frac{\sqrt{\bfe}}{2}, \frac{\sqrt{\bfe}}{2})$. The strict positivity of the cosine function in \eqref{eq:integral_lb2} follows if $| \Theta | < \frac{1}{2} < \frac{\pi}{2}$.  which is insured by
\beq\label{eq:positive1}
| \Theta| < |x_2 s| + |x_3| \sqrt{\bfe - s^2} < t_E \left( \frac{1 + \sqrt{3}) \sqrt{\bfe}}{2} \right)  < t_E \left( \frac{3 \sqrt{\bfe}}{2} \right)
 < \frac{1}{2} < \frac{\pi}{2}.
 \eeq
  Since we also require $t_E < \frac{1}{4}$, we get the condition
 \beq\label{eq:cond1}
 t_E \leq \min \left( \frac{1}{3 \sqrt{\bfe} } ,  \frac{1}{4} \right) .
 \eeq
 For $t_E$ satisfying \eqref{eq:cond1}, and inequality \eqref{eq:integral_lb2}, we get the lower bound,
 \beq\label{eq:integral_lb3}
| \langle \chi_{B_{t_E}(x_0)},  \psi_E \rangle | \geq \frac{\sqrt{2}}{2} | B_{t_E}(x_0)| \cos ( \frac{1}{2} ) .
\eeq
 \end{proof}


\subsection{Proof of Theorem \ref{thm:loc_deloc1}}

Let $I := [I_-, I_+] \subset ( \pi^2 , \infty)$. We use \eqref{eq:cond1} to determine the radius of the ball $B_t(x_0)$ for which the lower bound \eqref{eq:integral_lb1} for $\psi_E$ holds for $E \in I$.
If $0 < I_- < \pi^2 + \frac{4}{9}$, we may take any $0< t_I < \frac{1}{2}$. 
If $I_- > \pi^2 + \frac{4}{9}$, we take
\beq\label{eq:index_int1}
t_I :=  \frac{1}{3 \sqrt{I_-  - \pi^2}} ,
\eeq
so that $t_I < \frac{1}{2}$. 
In this way, the lower bound \eqref{eq:integral_lb1} with $t = t_I$ guarantees the lower bound for all $E \in I$:
\bea\label{eq:integral_lb4}
| \langle \chi_{B_{t_I}(x_0)},  \psi_E \rangle | & \geq & \frac{\eta }{2} | B_{t_I}(x_0)| \sin (\pi (\frac{1}{2} - t_I))   \nonumber \\
  & \geq & \frac{\sqrt{2}}{81}
   \left( \frac{\pi^4}{ ({I_- - \pi^2} )^\frac{3}{2}} \right) .
\eea


From \eqref{eq:mm_final1}, and bounds \eqref{eq:norm_ub1} and \eqref{eq:integral_lb2}, we find
\bea\label{defmm1}
  \lefteqn{ \mm_{P_\omega (I)\chi_{B_{t_I}(x_0)}, q}(T) }  \nonumber \\
   & = & \ov{T}\int_{0}^{\infty} e^{-\frac{t}{T}}\,\langle P_\omega (I) \chi_{B_{t_I}(x_0)} ,
   e^{iH_Vt} \varphi_q(X) e^{-iH_Vt}\, P_\omega (I)\chi_{B_{t_I}(x_0)} \rangle\;dt  \nonumber \\
 & \geq &
 \frac{1}{2 \pi T}  \int_I \| \varphi_q^{\frac{1}{2}} R_H( E + i \frac{1}{2T}) P_\omega (I) \chi_{B_{t_I}(x_0)} \|^2 ~dE      \nonumber \\
 & \geq & C(t_I, I, q) T ,
\eea
for $t_I$ defined in \eqref{eq:index_int1} ($t_I < \frac{1}{2}$), and $q > {3}$.
The constant $C(t,I,q)$ is given by
\beq\label{eq:cnst1}
C(t_I,I, q) := C_0  \left(\frac{q-3}{q} \right)
\left( \frac{\pi^6  }{(I_+ - \pi^2)^3}  \right) |I| > 0.
 \eeq
The constant $C_0$ is numerical. 

%


\begin{appendices}

\section{Proof of the Combes-Thomas estimate for $\delta$-potentials}\label{app: ct_proof1}
\setcounter{equation}{0}

In this appendix, we give the proof of the Combes-Thomas estimate for RDM, Theorem \ref{thm:ct_1}. The matrix $[\Gamma(z, \omega)]_{ij}$, for $i,j \in \Z^3$, is defined in \eqref{eq:kernel1}. This matrix  plays an essential role in defining point interaction Hamiltonians \eqref{eq:rdm1}. We restrict ourselves to the three-dimensional case.

We begin by studying the $\Gamma$-matrix defined in \eqref{eq:kernel1} that we recall here. The infinity-by-infinite matrix kernel $\Gamma_{}(z, \omega)$ is the random matrix on $\ell^2 ( \Z^3)$ given by
\beq\label{eq:kernel1.1}
[\Gamma (z, \omega)]_{ij} = \left( \frac{1}{\omega_{j}} - e(z) \right) \delta_{ij} - G_0(i,j; z)(1 - \delta_{ij} ) ,
\eeq
where $e(z) = \frac{\sqrt{z}}{4\pi}$, the square root defined with the principal branch.

The off-diagonal terms of $\Gamma (z, \omega)$ are given by the negative of the off-diagonal terms of the resolvent of the Laplacian $G_0(i,j;z)$, with $i \neq j$, $i, j \in \Z^3$. We note that for the off-diagonal terms of $\Gamma (z, \omega)$, we have
\beq\label{eq:decomp1}
\Ip ~ \Gamma (z, \omega)_{ij} = - \Ip ~ G_0(i,j;z), ~~~ i \neq j.
\eeq

In section \ref{subsec:dissipative1},  we give a general result about dissipative operators and apply it to the $\Gamma$-matrix in section \ref{subsec:Gamma1}.
In section \ref{subsec:exp_decay1}, we prove a general result about the exponential decay of the inverse of a matrix, given that the matrix itself decays exponentially.
We apply this to the $\Gamma$-matrix in section \ref{subsec:Gamma2}. We prove Theorem \ref{thm:ct_1} in section \ref{subsec:ct_main1}.


\subsection{Some estimates on dissipative operators}\leavevmode
\label{subsec:dissipative1}

\noindent We call an operator $T$ on a complex Hilbert space $\mathcal{H} $ dissipative if ${\rm Im } \,\langle \varphi,A\varphi \rangle\geq 0 $ for all $\varphi\in \mathcal{D}(T) $.
\begin{thm}\label{thm:inverse}
   Let $\mathcal{H} $ be a complex Hilbert space and $B, C$ self-adjoint operators on $\mathcal{H} $. If $C$ is bounded and
   \begin{align}
      C~\geq~\kappa~>~0\,,
   \end{align}
   then the operator
   \begin{align}
      A&~:=~B~+i\,C\\
      A& : {D}(B)\,\to\,\mathcal{H}
   \end{align}
   is bijective and
   \begin{align}
       \|{A^{-1}}\|~\leq~\frac{1}{\kappa}
   \end{align}
\end{thm}

%

\medskip

\begin{proof} 
 For any $\varphi\in  \mathcal{D}(B)$, we estimate
 \begin{align}
    |\;\langle \varphi, A\, \varphi\rangle\;|~&=~|\langle \varphi, B\, \varphi \rangle~+~i\,\langle \varphi, C\, \varphi \rangle|  \nonumber \\
    &\geq~|\;\langle \varphi, C\, \varphi \rangle\;|    \nonumber \\
    &\geq~\kappa\;\langle \varphi,\varphi \rangle\,,
 \end{align}
 where we used that both $\langle \varphi, B\, \varphi \rangle $ and $\langle \varphi, C\, \varphi \rangle $ are real.
Consequently,
 \begin{align}\label{eq:Aphi}
    \| A \varphi \|~\geq ~ \kappa\,\| \varphi \|\,.
 \end{align}
 It follows that $A$ is injective. Moreover,  the same argument shows that
 $A^{*}=B-iC$, satisfies $\| A^* \varphi\| \geq \kappa \| \varphi\|$.
We conclude that
 \begin{align}
    {\rm Ran} \, (A)^{\bot}~=~\textnormal{ker}(A^{*})~= \{ 0 \} .
 \end{align}
Hence, ${\rm Ran} \, (A)$ is dense in $\mathcal{H} $.
The operator $A$ is closed. Since $A$ satisfies \eqref{eq:Aphi},
it follows that ${\rm Ran} \, (A)=\mathcal{H}$.
With $\psi := A^{-1} \xi$, for $\xi  \in {\rm Ran} \, (A)=\mathcal{H}$, we estimate
\begin{align}
 \| \xi \| =  \| A \psi \| \geq \kappa \| \psi \| = \kappa \| A^{-1} \xi \|,
\end{align}
by \eqref{eq:Aphi}, so that
$$
\|A^{-1} \xi \| \leq \frac{1}{\kappa} \| \xi\|,
$$
for all $\xi \in \mathcal{H}$.
\end{proof}

\medskip

\begin{cor}
   Under the assumptions of Theorem \ref{thm:inverse}
   \begin{align}
      \sigma(A)~\subset~\{z=x+iy\mid ~x,y \in \R,  y\geq \kappa\}
   \end{align}
\end{cor}

\begin{proof}  
  For $z=x+iy$ with  $y<\kappa$ set
  \begin{align}
  A_{z}=A-z=(B-x)+\,i\,(C-y) \,.
  \end{align}
   By Theorem \ref{thm:inverse} we have
  $0\not\in\sigma(A_{z})$, hence $z=x+iy\not\in\sigma(A)$.
\end{proof}


\medskip


\subsection{Application to the $\Gamma$-matrix: Dissipative property}\label{subsec:Gamma1}\leavevmode

\noindent The goal of this section is to prove that $\Gamma (z, \omega)$ is boundedly invertible, using Theorem \ref{thm:inverse}. 
As usual, we set $A_{ij}=\langle \delta_{i},A\,\delta_{j} \rangle $ where $\{ \delta_{n} \}$ is the standard basis in $\ell^{2}(\Z^3)$.

We begin with a slightly more general construction which handles the case of `missing' sites, i.\,e. cases where $\omega_{n}=0 $ is allowed in the formal expression \eqref{eq:rdm1} (see Remark \ref{rem:omega0}).
Let $\Lambda$ be an arbitrary subset of the lattice $\Z^{3}$. Let  $\{ \omega_{n}, n\in\Lambda \}$ be a sequence of nonzero real numbers, and for $n,m\in\Lambda$,  we set
\beq\label{eq:kernel2}
[\Gamma(z, \omega)]_{nm} = \left(\frac{1}{\omega_{n}} -  i\,e(z) \right) \delta_{nm} - G_0(n,m; z)(1 - \delta_{nm} ) ,
\eeq
on $\ell^{2}(\Lambda)$. The parameter $e(z) :=\frac{\sqrt{z}}{4\pi} $, where the square root is taken with the principle branch. We first find a convenient expression for $\Ip ~ \Gamma ( z, \omega)$.

The matrix $G_0(n,m; z)$ is the kernel of the resolvent of $(H_{0}-z)^{-1}$ in $\R^{3} $ restricted to points in $\Lambda \times \Lambda$. That is, for $x \neq y \in \R^3$,  we have,
\begin{align}\label{eq:green11}
   G_{0}(x,y; z)~=~\frac{1}{4\pi}\frac{e^{i\sqrt{z}\| x-y \|}}{\| x-y \|}~
   =~\lim_{N\to\infty}\frac{1}{(2\pi)^{3}}\int_{\|p\|\leq N}\frac{e^{ip \cdot (x-y)}}{p^{2}-z}\,d^3p  .
\end{align}
For $x=y$, the above expression diverges. Observe, however, that
\bea\label{eq:im_trick1}
\textnormal{Im}{\int_{\|p \|\leq N}\frac{1}{p^{2}-z}\,d^3 p} & = & \textnormal{Im}{\int_{\|p\|\leq N} \left(  \frac{1}{p^{2}-z} - \frac{1}{p^2} \right)\,d^3 p}
\nonumber \\
 & = & \textnormal{Im}  \left( z  \int_{\|p\|\leq N} \frac{1}{p^2 ( p^{2}-z) }  \,d^3p \right) .
 \eea
By a contour integration, we have
  \bea\label{eq:singular1}
  \lim_{N \to \infty}  \int_{\|p\|\leq N} \left(  \frac{1}{p^{2}-z} - \frac{1}{p^2} \right)\,d^3 p  &=&
 \lim_{N\to\infty} ~ {2\pi z}\,{\int_{-N}^N \frac{1}{p^{2}-z}\,dp} \nonumber \\
   = 2 \pi^2 i \sqrt{z} .
\eea
As a consequence, we have
\beq\label{eq:im_part1}
\lim_{N \to \infty} \frac{1}{(2 \pi)^3} \textnormal{Im}{\int_{\|p\|\leq N}\frac{1}{p^{2}-z}\,d^3 p}  =  \textnormal{Im} \left( \frac{i \sqrt{z}}{4 \pi} \right),
\eeq
and may write
\begin{align}
   \Ip\,\big[ \Gamma(z,\omega) \big]_{nm} ~=~ - \Ip \left( \frac{1}{(2\pi)^{3}}\int_{\R^{3}}
   \,\Big(\frac{e^{ip \cdot (n-m)}}{p^{2}-z}\Big)\,d^3p \right) ,
\end{align}
with the understanding that for $n=m$, the value is
$-\textnormal{Im} ( ie(z) )$ as in \eqref{eq:kernel2}.

\medskip

\begin{thm}\label{thm:gamma_bd1}
For $\kappa >  0 $, the operator $\Gamma (E + i\kappa, \omega)$ is a bounded operator-valued function on $\ell^2(\Z^3)$.
The real part  $ {\rm Re} ~\Gamma (E + i \kappa)$ is a real symmetric matrix, and for $\kappa >  0 $, we have
   \begin{align}
   -  \Ip  ~\Gamma(E + i \kappa, \omega)~\geq C (E) \,\kappa,
   \end{align}
   for some $C (E) >0$. For $E \neq \pi^2$, the constant $C(E)$ can be chosen to be $ \frac{(2 \pi)^3}{(E-\pi^{2})^{2}} $\,.
\end{thm}

\begin{rem}
The proof can be transferred to the more general case where $\Lambda$ is a subset of an arbitrary lattice $\mathcal{L} $ in $\R^{3} $. In particular, $\Lambda $ may be an arbitrary finite set.
\end{rem}

\medskip

\noindent
We have an important corollary of Theorem \ref{thm:gamma_bd1}.

\medskip

\begin{cor}\label{eq:gamma_inv_bd1}
The bounded operator $\Gamma (E + i \kappa, \omega)$ on $\ell^2 (\Z^3)$ has a bounded inverse for $\kappa >  0$ and satisfies
\beq\label{eq:inv2}
\| \Gamma (E + i \kappa, \omega)^{-1} \| \leq \frac{C(E)}{\kappa} .
\eeq
\end{cor}

\medskip

\noindent
We now present the proof of Theorem \ref{thm:gamma_bd1}.

\begin{nota}
We denote the Fourier Transform of a function $\varphi $ on $\Z^{3}$ (a locally compact Abelian group) by
    \begin{align}
      \cF\varphi(p)~=~\frac{1}{(2\pi)^{3}}\sum_{n\in \Z^{3}}e^{-in \cdot p}\,\varphi(n)\,.
   \end{align}
 $\cF\varphi$ is a function on the torus $\mathbb{T}=\R^{3}/\Z^{3}$. We may as well regard
 $\cF\varphi$ as a function on $[-\pi,\pi)^{3} $ or as a periodic function on $\R^{3} $ with period cell $[-\pi,\pi)^{3} $.
\end{nota}

\noindent
We recall two basic facts about this Fourier transform:

\begin{enumerate}
\item The map $ \mathcal{F}$ is a bijective mapping between $\ell^{2}(\Z^{3})$ and $L^{2}(\mathbb{T})\cong L^{2}\big([-\pi,\pi)^{3}\big) $  and
   \begin{align}
      \| \mathcal{F}\varphi \|_{L^{2}\big([-\pi,\pi)^{3}\big)}=\| \varphi \|_{\ell^{2}(\Z^{3})} \,.
   \end{align}

\medskip

\item The inverse of $\cF $  is given by
   \begin{align}
      \mathcal{F}^{-1}u(n)~=\int_{[-\pi,\pi)^{3} }\,e^{ip \cdot n}u(p)\,d^3 p
   \end{align}
\end{enumerate}


\medskip

 \begin{proof} (Theorem \ref{thm:gamma_bd1})
 We compute
 \begin{align}
    \Ip{(H-E-i\kappa)^{-1}}~&=~\frac{1}{2i}\Big((H-E-i\kappa)^{-1}-(H-E+i\kappa)^{-1}\Big) \nonumber \\
    &=~\kappa\,(H-E+i\kappa)^{-1}(H-E-i\kappa)^{-1}\,.
 \end{align}
 The operator $(H-E+i\kappa)^{-1}(H-E-i\kappa)^{-1} $ has the kernel (in $x$-space) given by
 \begin{align}
    &(H-E-i\kappa)^{-1}(H-E+i\kappa)^{-1}(x,y)  \nonumber \\
    =~&\frac{1}{(2\pi)^{3}}\,\int_{\R^{3}}\,\frac{e^{-i p \cdot (x-y)}}{(p^2-E-i\kappa)(p^2-E+i\kappa)}\,d^3p
 \end{align}
For $\varphi \in \ell^{2}(\Lambda)$, we define an extension $\widetilde{\varphi}$ of $\varphi$ to $\Z^3$ by setting, for and $n\in \Z^{3}$, 
   \begin{align}
      \widetilde{\varphi}(n)~ := ~\left\{
                                 \begin{array}{ll}
                                   \varphi(n), & \hbox{if }n\in\Lambda \\
                                   0, & \hbox{oherwise.}
                                 \end{array}
                               \right.
   \end{align}
We then have
   \begin{align}
      & - \langle \varphi, \Ip \left(\Gamma\left(E + i\kappa, \omega\right)\right) \varphi \rangle_{\ell^{2}(\Lambda)}   \nonumber ~\\
      =~&
      \kappa\sum_{n,m\in\Z^{3}} \overline{\widetilde{\varphi}(n)}\, \left( \frac{1}{(2\pi)^{3}}\int_{\R^{3}}
   \frac{e^{-i p \cdot  (m-n)}}{(p^2-E-i\kappa)(p^2-E+i\kappa)}\,d^3p \; \right)
    \widetilde{\varphi}(m)   \nonumber \\
   =~&(2\pi)^{3}\kappa\int_{\R^{3}}\frac{|
\cF{\widetilde{\varphi}}(p)|^{2}}{(p^2-E-i\kappa)(p^2-E+i\kappa)} ~d^3p
 \nonumber \\
   \geq~&(2\pi)^{3}\kappa\int_{p\in[-\pi,\pi)^{3}} \frac{|\cF{\widetilde{\varphi}}(p)|^{2}}{(p^{2}-E)^{2}+\kappa^{2}} ~d^3p \nonumber \\
   \geq~&(2\pi)^{3}\kappa ~ \inf_{x\in [-\pi,\pi]^{3}} \Big(\frac{1}{(x^{2}-E)^{2}+\kappa^{2}}\Big)\,\int_{[-\pi,\pi]^{3}}|\cF\widetilde{\varphi}(p)|^{2} ~d^3p  \nonumber \\
   =~&C\,\kappa\,\| \widetilde{\varphi} \|_{\ell^{2}(\Z^{3})}~
=~C\,\kappa\,\| \varphi \|_{\ell^{2}(\Lambda)}  ,
   \end{align}
where the constant $C$ is
$$
C :=    (2\pi)^{3} ~ \inf_{x\in [-\pi,\pi]^{3}} \Big(\frac{1}{(x^{2}-E)^{2}+\kappa^{2}}\Big)  \geq \frac{(2\pi)^{3}}{({\pi}^{2}-E)^{2}} > 0.
$$

\end{proof}

\medskip

\subsection{Exponential decay of the inverse}\label{subsec:exp_decay1}

For the rest of this section we consider operators on $\lt(\Lambda)$
with $\Lambda\subset \Z^{d}$. In the following, $A$ will be a closed operator on $\lt(\Lambda) $.
We always assume that the standard basis $\{ \delta_{n}, n\in\Lambda \}$ belongs to the domain of $A $.  As usual we set
\begin{align}
   A_{nm}~=~\langle \delta_{n}, A\delta_{m} \rangle\,.
\end{align}
We are most interested in the case when $\Lambda$ is infinite , in particular, when $\Lambda = \Z^d$. The results hold for finite $\Lambda $ but are not of much interest.

\begin{thm}\label{thm:exp_decay1}
Suppose $\Lambda $ is a subset of $\Z^{d} $ and $A$ is an invertible closed operator on $\lt(\Lambda) $ with $\| A^{-1} \|\leq\rho $.
If there exist finite constants $\gamma>0 $ and $0< C_{0}<\infty $ so that
\begin{align}\label{eq:exp_decayA}
   |A_{nm}|~\leq~C_{0}\,e^{-\gamma\| n-m \|}\qquad\text{for all }n\not=m\,,
\end{align}
Then, there exists a constant $\mu_0 > 0$, depending on $\rho$ and $\gamma$ (see \eqref{eq:mu_bd1}), so that for any $0 < \mu <\mu_0$, we have
\begin{align}
   |{(A^{-1})}_{nm}|~\leq~ 2 \rho \,e^{-\mu\| n-m \|}\qquad\text{for all }n,m\,.
\end{align}
\end{thm}

\medskip

\begin{rem}
The proof is an adaption of \cite[Appendix II]{Aizenman},  see also \cite[section 10.3]{aw1} and \cite{Invitation}.
\end{rem}

\medskip

\begin{proof}
1. For a real number $\mu>0 $ to be determined later and $k\in\Lambda $,  we define the multiplication operator
\begin{align}\label{eq:exp_multi1}
   F_{k} u(n)~=~e^{\mu\| n-k \|}u(n)\,.
\end{align}
We easily compute
\begin{align}
   \big({F_{k}}^{-1}AF_{k}\big)_{nm}~&=~e^{-\mu\|n-k  \|}\,A_{nm}\,e^{\mu\|m-k  \|} , \label{eq:conj_A1} 
   \end{align}
 and
 \beq
 \big({F_{k}}^{-1}A^{-1}F_{k}\big)_{nm}~=~e^{-\mu\|n-k  \|}\,(A^{-1})_{nm}\,e^{\mu\|m-k  \|}  .  \label{eq:conj_A2}
\eeq 
It follows from setting $n=k$ in \eqref{eq:conj_A2} that
\begin{align}\label{eq:ubAinv1}
   |(A^{-1})_{nm}|~&=~e^{-\mu\| n-m \|}\,|(F_{n}^{-1}A^{-1}F_{n})_{nm}| \nonumber \\
&\leq~ e^{-\mu\| n-m \|}\,\|(F_{n}^{-1}A^{-1}F_{n})\|
\end{align}
Furthermore, from the resolvent formula applied to $F_n^{-1}A F_n$ and $A$ we obtain
\begin{align} \label{eq:conj_Ainv1}
   (F_{n}^{-1}AF_{n})^{-1} = F_{n}^{-1}A^{-1}F_{n} = A^{-1}-(F_{n}^{-1}AF_{n})^{-1}(F_{n}^{-1}AF_{n}-A)A^{-1},
\end{align}
from which it follows that
\begin{align} \label{eq:conj_Ainv2}
   (F_{n}^{-1}A^{-1}F_{n})~=~A^{-1}\,\Big(1+(F_{n}^{-1}AF_{n}-A)A^{-1}\Big)^{-1}
\end{align}
provided the inverse of $(1+(F_{n}^{-1}AF_{n}-A)A^{-1})$ exists. This is the case if $\|(F_{n}^{-1}AF_{n}-A)A^{-1}  \|<1 $. 

\noindent
2. To prove that $\|(F_{n}^{-1}AF_{n}-A)A^{-1}  \|<1$, we first bound $\|F_{n}^{-1}AF_{n}-A  \|$.  From \eqref{eq:exp_multi1} and \eqref{eq:conj_A1}, and any $k \in \Lambda$, we estimate for 
\begin{align}
  \left| \sum_{m\in\Lambda}\Big((F_{n}^{-1}AF_{n})_{km}-A_{km}\Big) \right| &
\leq~\sum_{m\in\Lambda, m\not= n} \Big(e^{-\mu\|k-n  \|}e^{+\mu\|m-n  \|}-1\Big)\,|A_{kn}|\label{eq:estsum1}
\end{align}
Observing that $|e^{a}-1|=|a\int_{0}^{1} e^{at}dt|\leq |a| e^{|a|}$, and using the decay hypothesis \eqref{eq:exp_decayA}, we find an upper bound for \eqref{eq:estsum1}:
\begin{align}
\lefteqn{  \left| \sum_{m\in\Lambda}\Big((F_{n}^{-1}AF_{n})_{km}-A_{km}\Big) \right| } \nonumber \\
  &\leq~\mu \sum_{m\in\Lambda,m\not=n}\Big( | \| k-n \||-\| m-n \| | e^{\mu | \| k-n \||-\| m-n \| | }\Big)e^{-\gamma\| k-m \|}   \nonumber \\
&\leq~\mu\sum_{m\in\Lambda}\,\| k-m \| e^{-(\gamma-\mu)\| k-m \|} ,
\end{align}
where we used $\big|\| k-n \|-\|m-n \|\big|\leq\| k-m \| $.
Consequently, if $\mu$ satisfies
\begin{align}\label{eq:mu_bd1}
   \mu~\leq \mu_0 := \min \left( \frac{\gamma}{2}, \big(2\rho\sum_{m}\|m\|e^{\frac{\gamma}{2}\| m \|}\big)^{-1} \right),
\end{align}
then for all $m,k\in\Lambda$,
\begin{align}
   a_{1}~=~|\sum_{m\in\Lambda}\Big((F_{n}^{-1}AF_{n})_{km}-A_{km}\Big)|\,\rho~\leq~\frac{1}{2}  .
\end{align}
In the same way one proves
\begin{align}
   a_{2}~=~|\sum_{k\in\Lambda}\Big((F_{n}^{-1}AF_{n})_{km}-A_{km}\Big)|\,\rho~\leq~\frac{1}{2} ,
\end{align}
for all $n$ and $k $.
By the Holmgren bound \cite[Lemma C.3]{aghkh}, we conclude that
\begin{align}
   \|(F_{n}^{-1}AF_{n}-A)A^{-1}  \|~&\leq~\rho\,\|(F_{n}^{-1}AF_{n}-A) \|\\
    &\leq~a_{1}^{\frac{1}{2}}\,a_{2}^{\frac{1}{2}}~=~\frac{1}{2}  .
\end{align}

\noindent
3. It follows that $1+(F_{n}^{-1}AF_{n}-A)A^{-1} $ is invertible with
$$
\|( 1+(F_{n}^{-1}AF_{n}-A)A^{-1})^{-1}  \| \leq {2}.
$$
Consequently,  from \eqref{eq:ubAinv1}, the assumed bound on $\| A^{-1} \| \leq \rho$, and \eqref{eq:conj_Ainv2}, we have
\begin{align}
   |(A^{-1})_{nm}|~&\leq~ 2 \rho\,e^{-\mu\| n-m \|} , 
\end{align}
for any $\mu$ satisfying \eqref{eq:mu_bd1}.
\end{proof}

\medskip


\subsection{Application to the $\Gamma$-matrix: Exponential decay}
\label{subsec:Gamma2}

We apply Theorem \ref{thm:exp_decay1} to the matrix $\Gamma (z,\omega)$ on $\ell^2 (\Z^3)$ as defined in \eqref{eq:kernel2}. From Theorem \ref{thm:gamma_bd1} we have that
\beq\label{eq:gamma_lb1}
- \Ip ~\Gamma (E + i \kappa, \omega) \geq C(E) \kappa > 0 ,
\eeq
for $\kappa > 0$. The constant $C(E)$ is defined for $E > \pi^2$ by:
\beq\label{eq:gamma_cnst1}
C(E) = \frac{(2 \pi)^3}{(E - \pi^2)^2}.
\eeq
From Theorem \ref{thm:inverse}, the dissipative operator $- \Gamma(z,\omega)$, $\Ip ~ z = \kappa > 0$, is boundedly invertible and
\beq\label{eq:gamma_inv1}
\| \Gamma(z,\omega)^{-1} \| \leq  \frac{1}{ C(E) \kappa}.
\eeq
From \eqref{eq:decomp1}, and the Combes-Thomas estimate for the Green's function of $H_0$ as in \eqref{eq:green12}, 
we have for $i \neq j$,
\beq\label{eq:exp_decay1}
| \Gamma(z, \omega)_{ij}| \leq |G_0(i,j;z)| \leq \frac{C}{d(z)} e^{\kappa(z)\|i - j\|}.
\eeq
Hence, $\Gamma(z,\omega)$ satisfies the hypotheses of Theorem \ref{thm:exp_decay1}.
As a consequence, we have
\beq\label{eq:gamma_inv2}
| ( \Gamma(z,\omega)^{-1})_{nm} | \leq  \frac{2}{ C(E) \kappa} e^{- \mu \| n - m\|}.
\eeq
where the constant $C(E)$ is given in \eqref{eq:gamma_cnst1} and for any $\mu$ satisfying \eqref{eq:mu_bd1}.



\subsection{Proof of Theorem \ref{thm:ct_1} }\label{subsec:ct_main1}

We use the exponential decay estimate \eqref{eq:gamma_inv2} to prove the boundedness of the interaction term on the left side of \eqref{eq:interact_term1}, proving the first part of Theorem \ref{thm:ct_1}.
If $x \in \R^3 \backslash \Z^3$, with $d(x)={\rm dist} ~(x, \Z^3)  > 0$, it follows from \eqref{eq:green11} that the Green's function $G_0$ satisfies
\beq\label{eq:green12}
| G_0(x,m; z)|  \leq \frac{1}{4 \pi d(x)} e^{- \tau(z) \|x-m \|}, ~~~ m \in \Z^3  ,
\eeq
where $\tau := \Ip ~ \sqrt{ E + i \kappa}$, for $z=E+i\kappa$ wih $\kappa > 0$.
In addition, we have the following lemma.

\begin{lemma}\label{lemma:double_sum1}
Suppose two arrays $\{ a_{mn} \}$ and $\{b_{mn} \}$, for $m,n \in \Z^3$, satisfy the exponential decay bounds
\beq\label{eq:bdd1}
| x_{mn}| \leq C e^{- \gamma \|n - m \|}  ,   ~~~x = a, {\rm and} ~~ x=b,
\eeq
for finite constants $C > 0$ and $\gamma > 0$.  Then, if $c_{mn} := \sum_{j \in \Z^3} a_{mj} b_{jn}$, there is a finite constant $\widetilde{C} > 0$ so that
\beq\label{eq:bdd1}
| c_{mn}| \leq \widetilde{C} e^{- \frac{\gamma}{2} \|n - m \|}  \,.
\eeq
$\widetilde{C}$ can be chosen as  $C^2 \gamma^{-3}$.
\end{lemma}

\begin{proof}
By direct calculation, we obtain
\bea\label{eq:prod_sum1}
|c_{mn}| &\leq&  C^2 \sum_{j \in \Z^3} e^{- {\gamma} \|j \|}  e^{- {\gamma} \|j - (n - m) \|}   \nonumber \\
 & \leq & C^2 \sum_{j \in \Z^3; \|j-(m-n)\| \geq \frac{\|m-n\|}{2}} e^{- {\gamma} 
 \|j\|}  e^{- {\gamma} \|j - (n - m)\|}
  \nonumber \\
  & & +~  C^2 \sum_{j \in \Z^3;\|j-(m-n)\| < \frac{\|m-n\|}{2} } e^{- {\gamma} \|j\|}  e^{- {\gamma} \|j - (n - m)\|}
  \nonumber \\
   & = & S_1 + S_2 .
   \eea
The sum $S_1$ is easily seen to be bounded as required and the constant
$\gamma^{-3}$ comes from $\sum_{j \in \Z^3} e^{- \gamma \|j\|}$. As for $S_2$, we note
\beq\label{eq:prod_sum2}
\|j\| = \|(m-n)-((m-n)-j)\| \geq \|j\| - \|(m-n)-j\|  \geq \frac{\|m-n\|}{2} .
\eeq
This establishes the lemma.
\end{proof}

In order to apply this to the Green's functions, we note the following proposition.

\begin{prop}\label{prop:green1} Let $C_0 := \{ x \in \R^3 ~|~ \|x_i\| \leq \frac{1}{2} \}$. For $\tau := \Ip ~ \sqrt{ E + i \kappa}$, for $\kappa > 0$ , we have

\begin{enumerate}

\item For any $x \in n + C_0$, with $d(x)={\rm dist} ~(x, \Z^3)  > 0$, and $m \in \Z^3$,
\beq\label{eq:green2}
| G_0(x,m; E+i\kappa)| \leq  \frac{ e^{\frac{1}{2} \tau}}{d(x)}\, e^{- \tau \|n-m\|}
\leq  \frac{ e^{\tau}}{d(x)}\, e^{- \tau \|x-m\|} .
\eeq

\medskip

\item
\beq\label{eq:green_ave1}
\int_{C_0}  | G_0(x+n,m; E+i\kappa)|~d^3 x  \leq 2 \pi  e^{\frac{1}{2} \tau}  e^{- \tau \|n-m \|}.
\eeq
\end{enumerate}

\end{prop}

Part (1) of Theorem \ref{thm:ct_1} follows from Lemma \ref{lemma:double_sum1} and the first part of Proposition \ref{prop:green1}. Part 2 of Theorem \ref{thm:ct_1} follows similarly from Lemma \ref{lemma:double_sum1} and part 2 of Proposition \ref{prop:green1}.


\section{Some basic resolvent estimates}\label{app:resolvent1}
\setcounter{equation}{0}

For any $E \in \R$ and any $\delta > 0$, we define the interval $I_\delta$ by
\beq\label{eq:interval1}
I_\delta := [ E - \delta, E + \delta].
\eeq
We denote by $I_\delta^c := \R \backslash I_\delta$. Let $H$ be a self-adjoint operator on a Hilbert space $\mathcal{H}$ with domain $D(H) \subset \mathcal{H}$.
For an interval $I \subset \R$, we let $P_I(H) = P_I$ denote the spectral projector for $H$ and interval $I$.

\medskip

\begin{lemma}\label{lemma:proj_bd1}
Let $H$ be a self-adjoint operator on a Hilbert space $\mathcal{H}$ and let $\psi \in D(H) \subset \mathcal{H}$. Suppose that there exists a finite $a > 0$ so that
\beq\label{eq:hyp1}
\| (H - E)\psi \| \leq a .
\eeq
We then have
\beq\label{eq:proj_bd1}
\| P_{I_\delta^c} \psi \| \leq \frac{a}{\delta}.
\eeq
\end{lemma}

\medskip

\begin{remark} Note that the constant $a$ depends on $\| \psi \|$.
\end{remark}

\begin{proof}
By the spectral theorem for $H$, we have
\bea\label{eq:norm1}
\| P_{I_\delta^c} \psi \|^2  & = & \int_\R \chi_{I_\delta^c} (s) d \mu_\psi^H (s)
  \nonumber \\
   & \leq & \frac{1}{\delta^2} \int_\R \chi_{I_\delta^2} (s) ( s - E)^2 d \mu_\psi^H (s)  \nonumber \\
    & \leq & \frac{1}{\delta^2} \| (H - E) \psi \|^2 ,
\eea
so the result follows from this and the bound \eqref{eq:hyp1}.
\end{proof}

\medskip

\begin{lemma}\label{lemma:projector_bd2}
Let $\epsilon > 0$. Under the hypotheses of Lemma \ref{lemma:proj_bd1}, we have
\beq\label{eq:proj_bd2}
\| R_H(E + i \epsilon) P_{I_\delta^c} \psi \| \leq \frac{a}{\delta^2} .
\eeq
\end{lemma}

\begin{proof}  We use the spectral representation as in \eqref{eq:norm1} modified by the insertion of the resolvent. Since $(s - E)^{-2} \leq \delta^{-2}$, for $s \in I_\delta^c$, as follows from \eqref{eq:interval1}, we obtain
\bea\label{eq:norm2}
\| R_H(E) P_{I_\delta^c} \psi \|^2  & = & \int_\R \chi_{I_\delta^c} (s) (s-E)^{-2} d \mu_\psi^H (s)
  \nonumber \\
   & \leq & \frac{1}{\delta^2} \int_\R \chi_{I_\delta^2} (s) d \mu_\psi^H (s)  \nonumber \\
    & \leq & \frac{1}{\delta^2} \| P_{I_\delta}  \psi \|^2  \nonumber \\
      & \leq & \frac{a^2}{\delta^4} ,
\eea
by Lemma \ref{lemma:proj_bd1}.
\end{proof}

\medskip

\end{appendices}



\begin{thebibliography}{10}

\bibitem{aghkh} S.\ Albeverio, F.\ Gesztesy, R.\ Hoegh-Krohn, H.\ Holden, {\it Solvable models in quantum mechanics}, Texts and Monographs in Physics, New York: Springer-Verlag, 1988. (Second Edition, AMS Chelsea Publishing, 2005.)

\bibitem{Aizenman} M.\ Aizenman, \emph{Localization at weak disorder: some elementary bounds}, Rev.\ Math.\ Phys. \textbf{6} (1994) 1163--1182; Special issue dedicated to Elliott H.\ Lieb.

\bibitem{aw1} M.\ Aizenman, S.\ Warzel, {\em Random operators, Disorder effects on quantum spectra and dynamics}, volume \textbf{168}, {Graduate Studies in Mathematics}, American Mathematical Society, Providence, RI, 2015.

\bibitem{bcm} J.\ M.\ Barbaroux, J.\ M.\ Combes, R.\ Montcho, \emph{Remarks on the relation between quantum dynamics and fractal spectra}, J.\ Math.\ Anal.\ Appl. \textbf{213}, (1997) 698-722.


\bibitem{CKL} M.\ Christ, A.\ Kiselev, Y.\ Last: \emph{Approximate eigenvectors and spectral theory}, Differential Equations and Mathematical Physics (Birmingham, AL, 1999), AMS/IP Stud.\ Adv.\ Math. \textbf{16},  Amer.\ Math.\ Soc., Providence, RI  (2000) 85 - 96.

\bibitem{combes} J.-M.\ Combes, \emph{Connection between quantum dynamics and spectral properties of time-evolution operators}, in Differential equations with applications to mathematical physics, 59–68, Academic Press, Boston, MA, 1993; MR1207148.

Springer Study Ed., Springer-Verlag, Berlin, 1987.

\bibitem{dBg2} S.\ De Bi\`evre, F.\ Germinet, Dynamical Localization for the Random Dimer \Schr Operator, Journal of Statistical Physics, \textbf{98}, Nos.\ 5-6 (2000).

\bibitem{DSS} D.\ Damanik, R.\ Sims, G.\ Stolz, Localization for discrete one-dimensional random word models, Journal of Functional Analysis \textbf{208} (2004) 423–445.



\bibitem{DKSB} M.\ Drabkin, W.\ Kirsch, H.\ Schultz-Baldes, \emph{Transport in the random Kronig-Penney model}, J.\ Math.\ Phys. \textbf{53} (2012), no.\ 12, 122109; arXiv:1207.0295v1.



\bibitem{EK} A\ Elgart, A.\  Klein,
\newblock \emph{Ground state energy of trimmed discrete {S}chr\"{o}dinger operators
  and localization for trimmed {A}nderson models},
\newblock {J.\ Spectr.\ Theory}, \textbf{4}(2) (2014) 391--413.

\bibitem{ES} A.\ Elgart, S.\ Sodin, \emph{The trimmed {A}nderson model at strong disorder: localisation and its breakup}, {J.\ Spectr.\ Theory}, \textbf{7}(1) (2017) 87-110.

\bibitem{GK-character}
F.\  Germinet, A.\  Klein,
\newblock \emph{A characterization of the {A}nderson metal-insulator transport
  transition},
\newblock {Duke Math.\ J.}, \textbf{124}(2):309--350, 2004.

\bibitem{GKS}  F.\ Germinet, A.\ Klein, and J.\ Schenker, \emph{Dynamical delocalization in random Landau Hamiltonians},
Ann.\ of Math. (2) \textbf{166} (2007), no.\ 1, 215–244.


\bibitem{guarneri} I.\ Guarneri, \emph{Spectral properties of quantum diffusion on lattices,} Europhys.\ Lett. \textbf{10} (1989), no. 2, 95–100; Europhys. Lett. 21 (1993), no. 7, 729–733

\bibitem{GuarneriSB}
I.\ Guarneri, H.\ Schulz-Baldes,
\newblock\emph{ Intermittent lower bound on quantum diffusion},
\newblock {Lett.\ Math.\ Phys.}, \textbf{49}(4):317--324, 1999.

\bibitem{hkk1} P.\ D.\ Hislop, W.\ Kirsch, M.\ Krishna, {\em Spectral and dynamical properties of random models with nonlocal singular interactions}, Math.\ Nachr.\ {\bf 278} No.\ 6, 627--664 (2005).

\bibitem{hkk2} P.\ D.\ Hislop, W.\ Kirsch, M.\ Krishna,  \emph{Eigenvalue statistics for \Schr operators with random point interactions on  $\R^d, ~~  d=1,2,3$},
J.\ Math.\ Phys. \textbf{61} (2020), no.\ 9, Paper No.\ 092103, 24 pp.

\bibitem{hkk3} P.\ D.\ Hislop, W.\ Kirsch, M.\ Krishna, \emph{Eigenfunctions and quantum transport with applications to trimmed \Schr operators}, J.\ Math.\ Phys. \textbf{65} (2024), no.\ 9, Paper No.\ 092103, 16 pp.
J.\ Math.\ Phys. (2024).



\bibitem{JSB}
S.\ Jitomirskaya,  H.\ Schulz-Baldes,
\newblock \emph{Upper bounds on wavepacket spreading for random {J}acobi matrices}, Comm.\ Math.\ Phys., \textbf{273}(3):601--618, 2007.

\bibitem{JSS} S.\ Jitomirskaya, H.\ Schulz-Baldes, and G.\ Stolz, \emph{Delocalization in random polymer models}, Comm.\ Math.\ Phys. \textbf{233} (2003), no.\ 1, 27-48.


\bibitem{Invitation}
W.\ Kirsch, \emph{An invitation to random {S}chr\"{o}dinger operators,}
\newblock In {\em Random {S}chr\"{o}dinger operators}, volume~\textbf{25} of {\em Panor.\   Synth\`eses}, pages 1--119. Soc. Math. France, Paris, 2008.
\newblock With an appendix by Fr\'{e}d\'{e}ric Klopp.

\bibitem{KiKr1}
W.\ Kirsch, M.\ Krishna,
\emph{Spectral statistics for {A}nderson models with sporadic potentials}, {J.\ Spectr.\ Theory},  \textbf{10}(2):581--597, 2020.

\bibitem{KiKr2}
W.\ Kirsch, M.~Krishna,
\newblock \emph{Localisation and {D}elocalisation for a {S}imple {Q}uantum {W}ave
  {G}uide with {R}andomness,}
\newblock {Ann.\ Henri Poincar\'{e}}, 23(11):4131--4148, 2022.


\bibitem{KL} A.\ Kiselev, Y.\ Last, \emph{Solutions, spectrum, and dynamics for Schr\"odinger operators on infinite domains}, Duke Math.\ J. \textbf{102} (2000), 12 - 150.

\bibitem{last} Y. Last, \emph{Quantum dynamics and decompositions of singular continuous spectra}, J.\ Funct.\ Anal. \textbf{142 }(1996), no. 2, 406–445.


%

\bibitem{CR}
C.~Rojas-Molina,
\newblock \emph{The Anderson model with missing sites},
\newblock {Oper.\ Matrices}, \textbf{8} (2014) 287-299.

%



%
%







\end{thebibliography}
\end{document}